\documentclass[aps, superscriptaddress,nofootinbib, twocolumn]{revtex4-1} 
\usepackage{graphicx}
\usepackage{amsmath}

\usepackage{amssymb}
\usepackage{bm}
\usepackage{enumerate}
\usepackage{amsthm}
\usepackage{tikz}
\usetikzlibrary{decorations.markings,decorations.pathreplacing,calc}
\usepackage{algorithm}
\usepackage{algpseudocode}
\usepackage{mathtools}
\mathtoolsset{showonlyrefs}

\usepackage{hyperref}
\usepackage{amsthm}
\usepackage{thm-restate}

\usepackage{subfig}

\newcommand{\be}{\begin{equation}}
\newcommand{\ee}{\end{equation}}
\newcommand{\ba}{\begin{array}}
\newcommand{\ea}{\end{array}}
\newcommand{\bea}{\begin{eqnarray}}
\newcommand{\eea}{\end{eqnarray}}

\newcommand{\cA}{{\cal A }}

\newcommand{\calI}{{\cal I }}
\newcommand{\calL}{{\cal L }}
\newcommand*{\cI}{\mathcal{I}}
\newcommand*{\cL}{\mathcal{L}}

\newcommand{\cS}{\mathcal{S}}

\newcommand{\ZZ}{\mathbb{Z}}

\newcommand{\la}{\langle}
\newcommand{\ra}{\rangle}

\newcommand{\HZn}{H_{\mathbb{Z}_n}}
\newcommand{\ghz}[1]{\mathsf{GHZ_{#1}}}

\newcommand{\cx}{\mathsf{CX}}
\newcommand{\cz}{\mathsf{CZ}}
\newcommand{\rz}{\mathsf{RZ}}

\newtheorem{theorem}{Theorem}

\newtheorem{lemma}{Lemma}

\newtheorem{corollary}[theorem]{Corollary}

\numberwithin{lemma}{section}

\newtheorem{corol}{Corollary}

\newcommand{\ket}[1]{|#1\rangle}
\newcommand{\bra}[1]{\langle#1|}

\newcommand{\proj}[1]{|#1\rangle\langle#1|}

\newtheorem{fact}{Fact}

\newcommand{\qaoap}[2]{{\mathsf{QAOA}_{#2}(#1)}} 

\newcommand*{\RQAOA}{\mathsf{RQAOA}}
\newcommand*{\QAOA}{\mathsf{QAOA}}
\DeclareMathOperator*{\argmax}{arg\,max}

\begin{document}
\title{Obstacles to State Preparation and Variational Optimization\\
from Symmetry Protection}
\author{Sergey Bravyi}
\affiliation{IBM T.J. Watson Research Center, Yorktown Heights, NY 10598, USA}
\author{Alexander Kliesch}
\affiliation{Zentrum Mathematik, Technical University of Munich, 85748 Garching, Germany}
\author{Robert Koenig}
\affiliation{Institute for Advanced Study \& Zentrum Mathematik, Technical University of Munich, 85748 Garching, Germany}
\author{Eugene Tang}
\affiliation{Institute for Quantum Information and Matter, Caltech,
Pasadena, CA 91125}

\begin{abstract}
Local Hamiltonians with topological quantum order
exhibit highly entangled ground states that
cannot be prepared by shallow quantum circuits.
Here, we show that this property may extend to all low-energy states in the presence of an on-site $\ZZ_2$ symmetry. This proves a version of the No Low-Energy Trivial States (NLTS) conjecture for a family of local Hamiltonians with symmetry protected topological order.
A surprising consequence of this result is that the Goemans-Williamson algorithm outperforms the Quantum Approximate Optimization Algorithm (QAOA) for certain instances of MaxCut, at any constant level. We argue that the locality and symmetry of QAOA severely limits its performance. To overcome these limitations, we propose a non-local version of QAOA, and give numerical evidence that it significantly outperforms standard QAOA for frustrated Ising models on random $3$-regular graphs.

\end{abstract}
\maketitle

Classifying topological phases of matter is amongst the main objectives of modern condensed matter physics~\cite{chiuetal16}.  Central to this program is the  characterization of entanglement structures 
that can 
emerge  in ground states of many-body systems. Of particular interest are
topologically non-trivial ground states~\cite{xiechengwen}.
Such states cannot be generated by a constant-depth
quantum circuit starting from a product state.
Non-trivial states exhibit complex, non-local entanglement properties and are thus expected to have highly non-classical features.  Remarkably, certain gapped local Hamiltonians
have non-trivial ground states. 
For example, preparing a ground state of Kitaev's toric code~\cite{kitaev2003fault} from a product state requires a  circuit depth growing at least polynomially in the system size using using nearest-neighbor gates~\cite{BravyiHastingsVerstraete06}, and logarithmically using non-local gates~\cite{aharonovtouati18}.

Going beyond ground states, a natural next question 
is whether there are local Hamiltonians with the property that  any low-energy state is non-trivial. Formalized by Freedman and Hastings~\cite{freedmanhastings13}, this is known as the No Low-energy Trivial States (NLTS) conjecture. To state it in detail, consider many-body systems composed of $n$~finite-dimensional subsystems -- assumed here to be  qubits for simplicity. A local Hamiltonian is  a sum of interaction terms acting non-trivially on $O(1)$ qubits each. We require that each term has operator norm $O(1)$ and each qubit is involved in $O(1)$ terms. The interaction terms may be long-range (no geometric locality is needed). It will be assumed that a local Hamiltonian $H_n$ as above is defined for each $n\in \calI$, where $\calI$ is some infinite set of system sizes. A family of local Hamiltonians $\{H_n\}_{n\in\calI}$ is said to have the NLTS property  
if there exists a constant $\epsilon>0$ and a function
$f\, : \, \ZZ_+\to \ZZ_+$ 
such that:\\
(1) $H_n$ has ground state energy $0$ for any $n\in \calI$,\\
(2) $\la 0^n|U^\dag H_n U|0^n\ra> \epsilon n$
for any depth-$d$ circuit $U$ composed of two-qubit gates
and for any $n\ge f(d)$, $n\in \calI$.\\
Here the circuit depth $d$ can be arbitrary. The conjecture is that the NLTS property holds for some family of local Hamiltonians.

The validity of the NLTS conjecture is a necessary condition for the stronger quantum PCP conjecture to hold~\cite{hastingsqpcp}: the latter posits that there are local Hamiltonians whose ground state energy is QMA-hard to approximate with an extensive error $\epsilon n$ for some constant $\epsilon>0$.

A proof of the NLTS conjecture is still outstanding. Although 
many natural families of Hamiltonians provably do not have the NLTS property (see~\cite{eldarharrow} for a comprehensive list), evidence for its validity has been provided by a number of related results. Ref.~\cite{freedmanhastings13}  constructs  Hamiltonians satisfying a certain one-sided NLTS property: these have excitations of two kinds (similar to the toric code), and low-energy states with no excitations of the first kind are non-trivial. The construction crucially relies on expander graphs as even the one-sided NLTS property does not hold for similar constructions on regular lattices~\cite{freedmanhastings13}. 

Eldar and Harrow~\cite{eldarharrow} construct families of local Hamiltonians (also based on expander graphs) such that any state whose reduced density operators on a constant fraction of sites coincides with that of a ground state is  non-trivial. This feature, called the No Low-Error Trivial States (NLETS) property, is clearly related to robustness of entanglement in the ground state with respect to erasure errors~\cite{aharonovaradvidick2013,nirkhevaziraniyuen2018}. The existence of Hamiltonians with the NLETS property is a necessary condition~\cite{eldarharrow} for the existence of good quantum LDPC codes, another central conjecture in quantum information.

Here we pursue a different approach to the NLTS conjecture by imposing an additional symmetry in the initial state as well as the preparation circuit. This mirrors similar considerations in the classification of topological phases, where the concept of symmetry-protected topological (SPT) phases~\cite{GuWenSPT09} has been extremely fruitful.  Indeed, the study of SPT equivalence classes of states, pioneered in~\cite{haldane83,aklt87}, has led to a complete classification of 1D~phases~\cite{GuWenSPT09,schuch11,pollmannetal10}, and also plays an essential role in measurement-based quantum computation~\cite{weiafflekraussendorf12,millermiyake15}.

For concreteness, we focus on the simplest case of onsite $\mathbb{Z}_2$-symmetry. A local Hamiltonian is said to be $\mathbb{Z}_2$-symmetric if all interaction terms commute with  $X^{\otimes n}$, where $X$ is the single-qubit Pauli-$X$ operator. Likewise, a quantum circuit $U$ acting on $n$ qubits is said to be $\mathbb{Z}_2$-symmetric if it obeys
\begin{align}
\label{eq:Ztwosym}
UX^{\otimes n} = X^{\otimes n} U\ .
\end{align}
We do not impose the symmetry on the individual gates of $U$, though this will naturally be the case in many interesting examples, such as the QAOA circuits considered below.
Finally, let us say that a state $\Psi$ of $n$~qubits is $\mathbb{Z}_2$-symmetric 
if $X^{\otimes n}\Psi = \pm \Psi$.
Our first result is a proof of the NLTS conjecture in the presence of
onsite  $\mathbb{Z}_2$-symmetry:
\begin{theorem}\label{thm:nlts}
There exist constants $\epsilon,c>0$ and  a family of $\mathbb{Z}_2$-symmetric local 
Hamiltonians $\{H_n\}_{n\in \calI}$
such that $H_n$ has ground state energy~$0$ for any $n\in \calI$ while
\begin{align}
\bra{\varphi}U^{\dag}H_nU\ket{\varphi}> \epsilon n  
 \label{eq:upperboundHamiltonianenergy}
\end{align}
for any $\mathbb{Z}_2$-symmetric depth-$d$ circuit $U$
composed of two-qubit gates,
any $\mathbb{Z}_2$-symmetric product state $\varphi$,
and any $n\ge 2^{cd}$, $n\in \calI$.
\end{theorem}

Our starting point to establish Theorem~\ref{thm:nlts} is  a fascinating
result by Eldar and Harrrow stated as Corollary~43 in~\cite{eldarharrow}.
It shows that the output distribution of a shallow quantum circuit
cannot assign a non-negligible probability to subsets of bit strings
that are separated far apart w.r.t.~the Hamming distance.
More precisely, define the distribution $p(x)=|\la x|U|\varphi\ra|^2$,
where $x\in \{0,1\}^n$. Given a subset $S\subseteq \{0,1\}^n$,
let  $p(S)=\sum_{x\in S} p(x)$. 
\begin{fact}[\cite{eldarharrow}]
For all subsets $S,S'\subseteq \{0,1\}^n$  one has
\be
\label{Corol43}
\mathrm{dist}(S,S')\le \frac{4 n^{1/2} 2^{3 d/2}}{\min{\{ p(S),p(S') \} }}.
\ee
\end{fact}
Here $\mathrm{dist}(S,S')$ is the Hamming distance,
i.e. the minimum number of bit flips required to get from $S$ to $S'$.
We emphasize that Eq.~\eqref{Corol43} holds for all depth-$d$
circuits $U$ and all product states $\varphi$ ($\mathbb{Z}_2$-symmetry is not needed).

Given a bit string~$x$ let $\overline{x}$ be the bit-wise negation of~$x$.
Note that  $p(x)=p(\overline{x})$ since $U\varphi$ is $\mathbb{Z}_2$-symmetric.
Choose $S$ and $S'$ as the sets of all $n$-bit strings with the Hamming weight 
$\le n/3$ and  $\ge 2n/3$ respectively.
Then $p(S')=p(S)$ and $\mathrm{dist}(S,S')=n/3$.
Eq.~\eqref{Corol43} gives
\be
\label{pi(S)}
p(S)\le 12 n^{-1/2} 2^{3 d/2}.
\ee
Our strategy is to choose the Hamiltonian $H_n$ such that low-energy states
of $H_n$ are concentrated on bit strings with the Hamming weight close to $0$ or $n$
such that $p(S)$ is non-negligible. Then Eq.~\eqref{pi(S)}  provides a logarithmic lower bound on
the depth $d$ for symmetric low-energy states.

Suppose $G=(V,E)$ is a graph with $n$ vertices.
The  Cheeger constant of $G$ is defined as 
\begin{align}
\label{Cheeger}
h(G)&=\min_{\substack{
S\subseteq V\\
0<|S|\leq n/2
}} \frac{|\partial S|}{|S|}\ ,
\end{align}
where $\partial S\subseteq E$ 
is the subset of edges that have exactly one endpoint in~$S$. 
Families of expander graphs are infinite collections of
bounded degree graphs
$\{G_n\}_{n\in \calI}$ whose Cheeger constant is lower bounded by a constant,
i.e.  $h(G_n)\ge h>0$ for all $n\in \calI$. Explicit constructions
of degree-$3$ expanders can be found in~\cite{MORGENSTERN199444}.
Fix a family of degree-$3$ expanders  $\{G_n\}_{n\in \calI}$ 
and define $H_n$ as the ferromagnetic Ising model on the graph $G_n$, i.e.
\be
\label{FIM}
H_n=\frac12 \sum_{(u,v)\in E} (I - Z_u Z_v).
\ee
Here $Z_u$ is the Pauli-$Z$ applied to a site $u$ and $I$ is the identity.
Clearly, $H_n$ is $\mathbb{Z}_2$-symmetric,
each term in $H_n$ acts on two qubits and each qubit is involved in three terms. 
Thus $\{H_n\}_{n\in \calI}$ is a family of local Hamiltonians.
$H_n$ has 
$\mathbb{Z}_2$-symmetric ground states 
$\frac1{\sqrt{2}} (|0^n\ra \pm |1^n\ra)$ with zero energy.
Given a bit string $x$, let $\mathsf{supp}(x)=\{ j\in [n]\, : \, x_j=1\}$ be the support of $x$.
From Eqs.~\eqref{Cheeger},\eqref{FIM} one gets
\be
\label{low_energy1}
\la x|H_n|x\ra = |\partial\,  \mathsf{supp}(x)|\ge h\cdot \min{\{|x|, n-|x|\}},
\ee
where $|x|$ is the Hamming weight of $x$. Assume
Eq.~\eqref{eq:upperboundHamiltonianenergy} is false.
Then  $p(x)$ is a  
low-energy distribution such that
$\sum_x p(x) \la x|H_n|x\ra \le \epsilon n$.
By Markov's inequality, $p(x)$ has a non-negligible
weight on low-energy basis states, 
\be
\label{low_energy2}
p(S_{\mathrm{low}})\ge 1/2, \qquad S_{\mathrm{low}}=\{x \, : \, \la x|H_n|x\ra \le 2\epsilon n \}.
\ee
By Eq.~\eqref{low_energy1}, $\min{\{|x|, n-|x|\}} \le 2n\epsilon h^{-1}$ for all $x\in S_{\mathrm{low}}$.
Choose $\epsilon=h/6$. 
Then $S_{\mathrm{low}}\subseteq S\cup S'$ and 
$p(S_{\mathrm{low}})\le p(S)+p(S')=2p(S)$.
Here the last equality uses the symmetry of $p(x)$.
By Eq.~\eqref{low_energy2}, $p(S)\ge 1/4$.
Combining this and Eq.~\eqref{pi(S)} one gets
$n^{1/2}\le 48 \cdot 2^{3d/2}$.
We conclude that Eq.~\eqref{eq:upperboundHamiltonianenergy}
holds whenever $n>48^2 \cdot 8^d$. This proves Theorem~\ref{thm:nlts}.

The Hamiltonians Eq.~\eqref{FIM} are diagonal in the computational basis and have product ground states $|0^n\ra$ and $|1^n\ra$.  The presence of the $\mathbb{Z}_2$-symmetry
is therefore essential: the same family of Hamiltonians do not exhibit NLTS without it.
In this sense, the NLTS property here behaves similarly to topological order 
in 1D systems which  only exists under symmetry protection~\cite{GuWenSPT09, schuch11}.

Theorem~\ref{thm:nlts} implies restrictions on the performance of variational quantum algorithms for combinatorial optimization. Recall that the Quantum Approximate Optimization Algorithm (QAOA)~\cite{farhi2014} seeks to approximate the maximum of a cost function $C:\{0,1\}^n\rightarrow\mathbb{R}$ by encoding it into a Hamiltonian~$H_n=\sum_{x\in \{0,1\}^n}C(x)\proj{x}$. It variationally optimizes the expected energy of $H_n$ over quantum states of the form~$U(\beta,\gamma)\ket{+^n}$, where 
\begin{align*}
U(\beta,\gamma)=\prod_{k=1}^p e^{i\beta_kB }e^{i\gamma_k H_n},
\end{align*} 
and where $B=\sum_{j=1}^n X_j$. 
The integer $p\ge 1$ is called the QAOA {\em level}. It controls non-locality of the variational circuit. 

A paradigmatic test case for QAOA is the MaxCut problem.
Given a graph $G_n=(V,E)$ with $n$ vertices,
the corresponding MaxCut Hamiltonian is defined by Eq.~\eqref{FIM}.
The maximum energy of $H_n$ coincides with the number of edges in the maximum cut
of $G_n$. Crucially,  the QAOA circuit $U(\beta,\gamma)$ 
with the Hamiltonian $H_n$
as well as the initial QAOA state $\ket{+^n}$ obey the 
$\mathbb{Z}_2$-symmetry property.
Furthermore, the circuit $U(\beta,\gamma)$ has depth
$O(Dp)$, where $D$ is the maximum vertex degree of $G_n$.
Specializing Theorem~\ref{thm:nlts} to bipartite
graphs we obtain an upper  
bound on the approximation ratio 
achieved by the level-$p$ QAOA circuits 
for the MaxCut cost function (see Appendix~A for a proof):
\begin{corol}
For every integer $D\ge 3$ there exists an
infinite family  of bipartite $D$-regular graphs $\{G_n\}_{n\in \calI}$
such that the Hamiltonians $H_n$
defined in Eq.~\eqref{FIM} obey
\begin{align}
\label{eq:mainclaimqaoa}
\frac1{|E|}
\bra{+^n} U^{-1}H_n U\ket{+^n}
\le \frac{5}{6} + \frac{\sqrt{D-1}}{3D}
\end{align}
for any level-$p$ QAOA circuit $U\equiv U(\beta,\gamma)$ 
as long as 
$p<(1/3\, \log_2n - 4)D^{-1}$.
\label{corol:QAOA1}
\end{corol}
Note that  any bipartite graph with a set of edges $E$ has maximum cut size $|E|$.
In this case, the left-hand side
of Eq.~\eqref{eq:mainclaimqaoa} coincides with the approximation ratio, i.e., the ratio between the expected value of the MaxCut cost function on the (optimal) level-$p$ variational state
and the maximum cut size.
Thus Corollary~\ref{corol:QAOA1}  provides an explicit upper bound on 
the approximation ratio achieved by level-$p$ QAOA. 
Such bounds were previously known only for $p=1$~\cite{farhi2014}. Statement~\eqref{eq:mainclaimqaoa} severely limits the performance of QAOA at any constant level~$p$, rigorously establishing a widely believed conjecture~\cite{hastings2019}: constant-level QAOA is inferior to the classical Goemans-Williamson algorithm for MaxCut, which achieves an approximation ratio of approximately $0.878$ on an arbitrary graph~\cite{goemanswilliamson1995}.
Indeed, the right-hand side of  Eq.~\eqref{eq:mainclaimqaoa} is approximately $5/6\approx 0.833$
for large vertex degree $D$.

QAOA circuits $U(\beta,\gamma)$  possess a  form of locality 
which is stronger than the one assumed in  Theorem~\ref{thm:nlts}.
Indeed, if $p$ and $D$ are constants,
the unitary
$U(\beta,\gamma)$  can be realized by a constant-depth circuit composed of  {\em nearest-neighbor} gates, i.e., the circuit is geometrically local. 

A natural question is whether more general bounds on the variational energy can be established for states generated by 
geometrically local
$\ZZ_2$-symmetric circuits. Of particular interest are graphs that lack the expansion property, such as regular lattices, where the arguments used in the proof of Theorem~\ref{thm:nlts} no longer apply.
A simple model of this type is the {\em ring of disagrees}~\cite{farhi2014}.
It describes the MaxCut problem on the cycle graph $\ZZ_n$. 

Quite recently, Ref.~\cite{mbengfaziosantoro2019} proved that the optimal approximation ratio achieved by level-$p$ QAOA for the ring of disagrees is bounded above by $(2p+1)/(2p+2)$ for all $p$ and conjectured that this bound is tight. 
Here we prove a version of this conjecture for  arbitrary geometrically local $\ZZ_2$-symmetric circuits.
To quantify the notion of geometric locality, let us say that 
a unitary $U$ acting on $n$ qubits located at vertices of the cycle graph
has range $R$ if the operator
$U^\dag Z_j U$ has support on the interval $[j-R,j+R]$
for any qubit $j$. For example,  the
level-$p$ QAOA circuit associated with the ring of disagrees has range $R=p$.

\begin{theorem}
\label{thm:1D}
Let $H_n$ be the ring of disagrees Hamiltonian,
\[
H_n=\frac12 \sum_{p\in \ZZ_n}(I-Z_p Z_{p+1}),
\]
where $n$ is even. Let $U$ be a $\ZZ_2$-symmetric unitary with range $R<n/4$.
Then 
\be
\label{cycle1}
\frac1n\la +^n|U^\dag H_n U|+^n\ra \le \frac{2R+1/2}{2R+1}.
\ee
This  bound  is tight whenever
$n$ is a multiple of $2R+1$.
\end{theorem}
Since one can always round $n$ to the nearest multiple of $2R+1$, the
bound Eq.~\eqref{cycle1} is tight for all $n$  up to corrections $O(1/n)$,
assuming that $R=O(1)$.

Let us first prove the upper bound Eq.~\eqref{cycle1}.  
Define $\overline{X}=(XI)^{\otimes n/2}$.
Then
\be
\label{cycle2}
\overline{X}H_n\overline{X} + H_n = nI
\ee
Let $V=\overline{X} U$. Note that $V$ is a $\ZZ_2$ symmetric circuit
with range $R$. Taking the expected value of Eq.~\eqref{cycle2} on the 
state $U|+^n\ra$, one infers that
Eq.~\eqref{cycle1} holds whenever
\be
\label{cycle3}
\frac1n \la +^n |V^\dag H_n V|+^n\ra \ge 1-\frac{2R+1/2}{2R+1} = \frac1{2(2R+1)}.
\ee
Thus it suffices to prove that Eq.~\eqref{cycle3} holds for
any $\ZZ_2$-symmetric range-$R$ circuit $V$.
For each $j,k\in \ZZ_n$ define
\[
\epsilon_{j,k} = \frac12\la +^n|V^\dag (I-Z_j Z_k)V|+^n\ra.
\]
Let $\mathrm{dist}(j,k)$ be the distance between $j$ and $k$
with respect to the cycle graph $\ZZ_n$. 
We claim that  
\be
\label{ZZ=0}
\epsilon_{j,k} =1/2 \quad \mbox{if} \quad  \mathrm{dist}(j,k)>2R.
\ee
Indeed, 
$\la +^n|V^\dag Z_i  V |+^n\ra=0$  for any qubit $i$
since $V|+^n\ra$ and $Z_iV|+^n\ra$ are 
eigenvectors of $X^{\otimes n}$
with eigenvalues~$1$ and $-1$.  Such eigenvectors
have to be orthogonal.
From $\mathrm{dist}(j,k)>2R$ one infers that $V^\dag Z_jV$ and $V^\dag Z_kV$
have disjoint support. Thus
\begin{align*}
\la +^n|V^\dag Z_j Z_k V |+^n\ra=\la +^n|(V^\dag Z_jV)(V^\dag Z_k V) |+^n\ra  \\
=\la +^n|V^\dag Z_j V |+^n\ra\cdot \la +^n|V^\dag Z_k V |+^n\ra=0.
\end{align*}
This proves Eq.~\eqref{ZZ=0}. 
Suppose one prepares the state $V|+^n\ra$ and measures a pair of qubits $j<k$
in the standard basis. 
Then $\epsilon_{j,k}$ is the probability that the measured values on qubits $j$ and $k$ disagree.
By the union bound,
\be
\label{ring1}
\epsilon_{j,k}\le \sum_{i=j}^{k-1} \epsilon_{i,i+1}.
\ee
Indeed, if qubits $j$ and $k$ disagree, at least one pair of consecutive
qubits located in the interval $[j,k]$ must disagree.
Set $k=j+2R+1$.
Then $\epsilon_{j,k}=1/2$ by Eq.~\eqref{ZZ=0}.
Take the expected value
of Eq.~\eqref{ring1} with respect to random uniform $j\in \ZZ_n$.
This gives 
\[
\frac12 \le \frac{2R+1}{n} \sum_{i\in \ZZ_n} \epsilon_{i,i+1} = \frac{2R+1}{n}\la +^n|V^\dag H_n V|+^n\ra
\]
proving Eq.~\eqref{cycle3}. In Appendix~B we construct a $\ZZ_2$-symmetric range-$R$ circuit~$U$ 
such that $U|+^n\ra$ is a tensor product of GHZ-like states on consecutive segments of $2R+1$ qubits. We show that such circuit saturates the upper bound Eq.~\eqref{cycle1}.
This completes the proof of Theorem~\ref{thm:1D}.

Concerns about limitations of QAOA have previously been voiced by Hastings~\cite{hastings2019} who showed analytically that certain local classical algorithms match the performance of level-$1$ QAOA for Ising-like cost functions with multi-spin interactions. Hastings also gave numerical evidence for the same phenomenon for MaxCut with $p=1$, and argued that this should extend to $p>1$~\cite{hastings2019}.
 
Motivated by these limitations, we propose a non-local modification of QAOA
which we call the recursive quantum approximate optimization algorithm (RQAOA).
To sketch the main ideas behind RQAOA consider  an Ising-like  Hamiltonian
\begin{equation}
\label{Ising1}
H_n=\sum_{(p,q)\in E} J_{p,q} Z_p Z_q
\end{equation}
defined on a  graph $G_n=(V,E)$ with $n$ vertices.
Here $J_{p,q}$ are arbitrary real coefficients. 
RQAOA aims to approximate the maximum energy
$\max_z \la z|H_n|z\ra$, where $z\in \{1,-1\}^n$. It consists of the following steps.

First, run the standard QAOA to maximize the expected value of $H_n$
on the  state $|\psi\ra=U(\beta,\gamma)|+^n\ra$.
For every edge $(j,k)\in E$ compute
$M_{j,k}=\la \psi^*|Z_j Z_k|\psi^*\ra$, where $\psi^*$
is the optimal variational state.

Next, find a pair of qubits $(i,j)\in E$ with the largest magnitude of $M_{i,j}$
(breaking ties arbitrarily).
The corresponding variables $Z_i$ and $Z_j$ are correlated if $M_{i,j}>0$
and anti-correlated if $M_{i,j}<0$.
Impose the constraint 
\be
\label{RQAOA1}
Z_j = \mathrm{sgn}{(M_{i,j})} Z_i
\ee
and substitute it into the Hamiltonian $H_n$ to eliminate the variable $Z_j$.
For example, a term $Z_j Z_k$ with $k\notin \{i,j\}$ gets mapped to $\mathrm{sgn}{(M_{i,j})}Z_i Z_k$.
The term $J_{i,j}Z_i Z_j$ gets mapped to a constant energy shift
$J_{i,j}\mathrm{sgn}{(M_{i,j})}$. All other terms remain unchanged.
This yields a new Ising Hamiltonian $H_{n-1}$ that depends on $n-1$ variables.
By construction, the maximum  energy of $H_{n-1}$
coincides with the maximum energy of $H_n$
over the subset of assignments  satisfying the  constraint Eq.~\eqref{RQAOA1}.

Finally, call RQAOA recursively to  maximize the expected value of  $H_{n-1}$.
Each recursion step  eliminates one variable from the cost function.
The recursion stops when the number of variables reaches
some specified threshold value $n_c\ll n$.
The remaining instance of the problem with $n_c$ variables
is then solved by a purely classical algorithm (for example, by a brute force method).
Thus the value of $n_c$ controls how the workload is distributed between quantum and classical computers. We describe a generalization of RQAOA applicable to Ising-like cost functions with multi-spin interactions in Appendix~C.

Imposing a constraint of the form~\eqref{RQAOA1} can be viewed as rounding correlations
among the variables $Z_i$ and $Z_j$. Indeed, the constraint demands that
these variables must be perfectly correlated or anti-correlated. This is analogous to rounding fractional solutions obtained by solving linear programming relaxations of combinatorial optimization problems. We note that reducing the size of a problem to the point that it can be solved optimally by brute force is a widely used and effective approach in combinatorial optimization.

\begin{figure*}[t]
\centerline{\includegraphics[height=5.2cm]{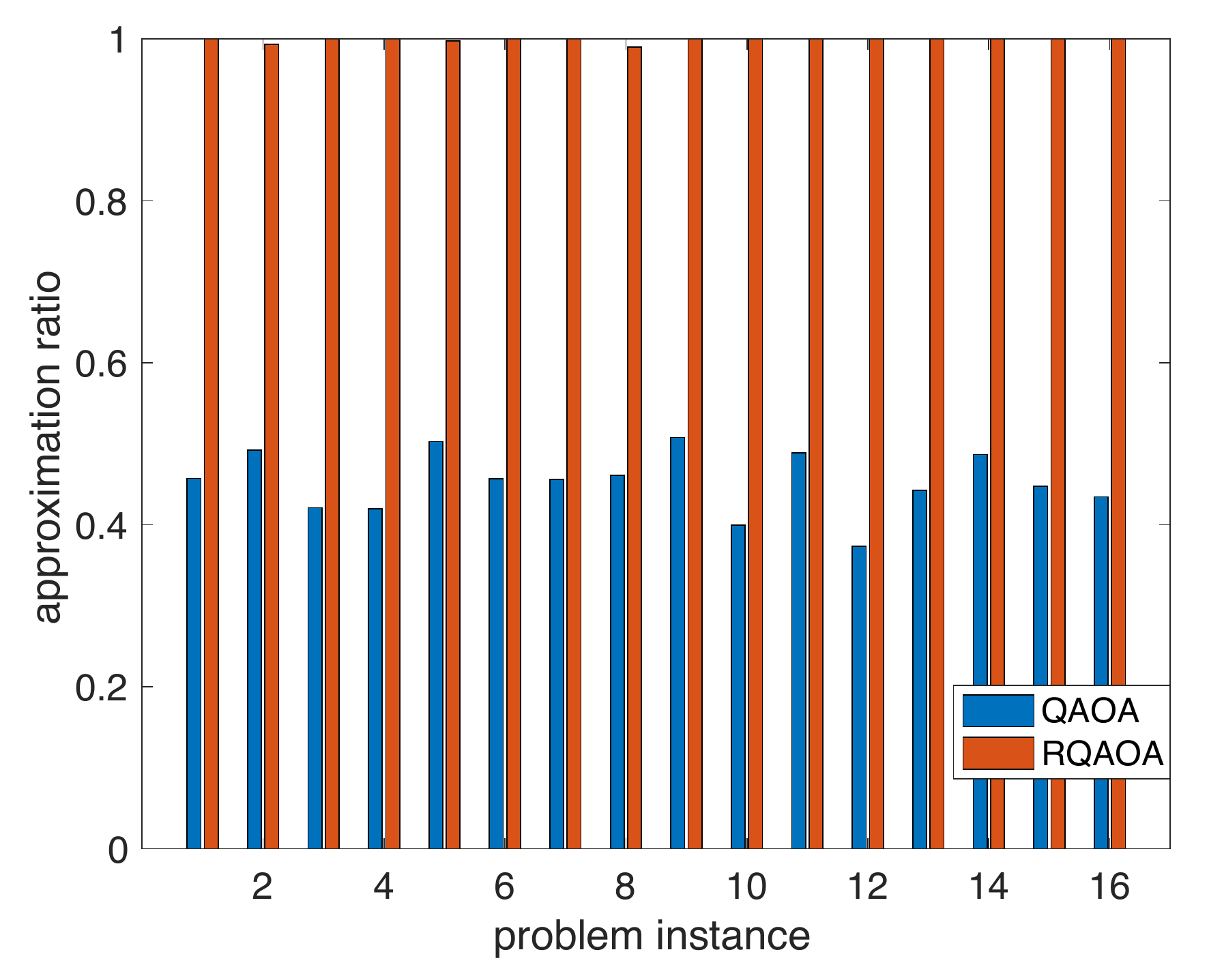} $\qquad$
\includegraphics[height=5.2cm]{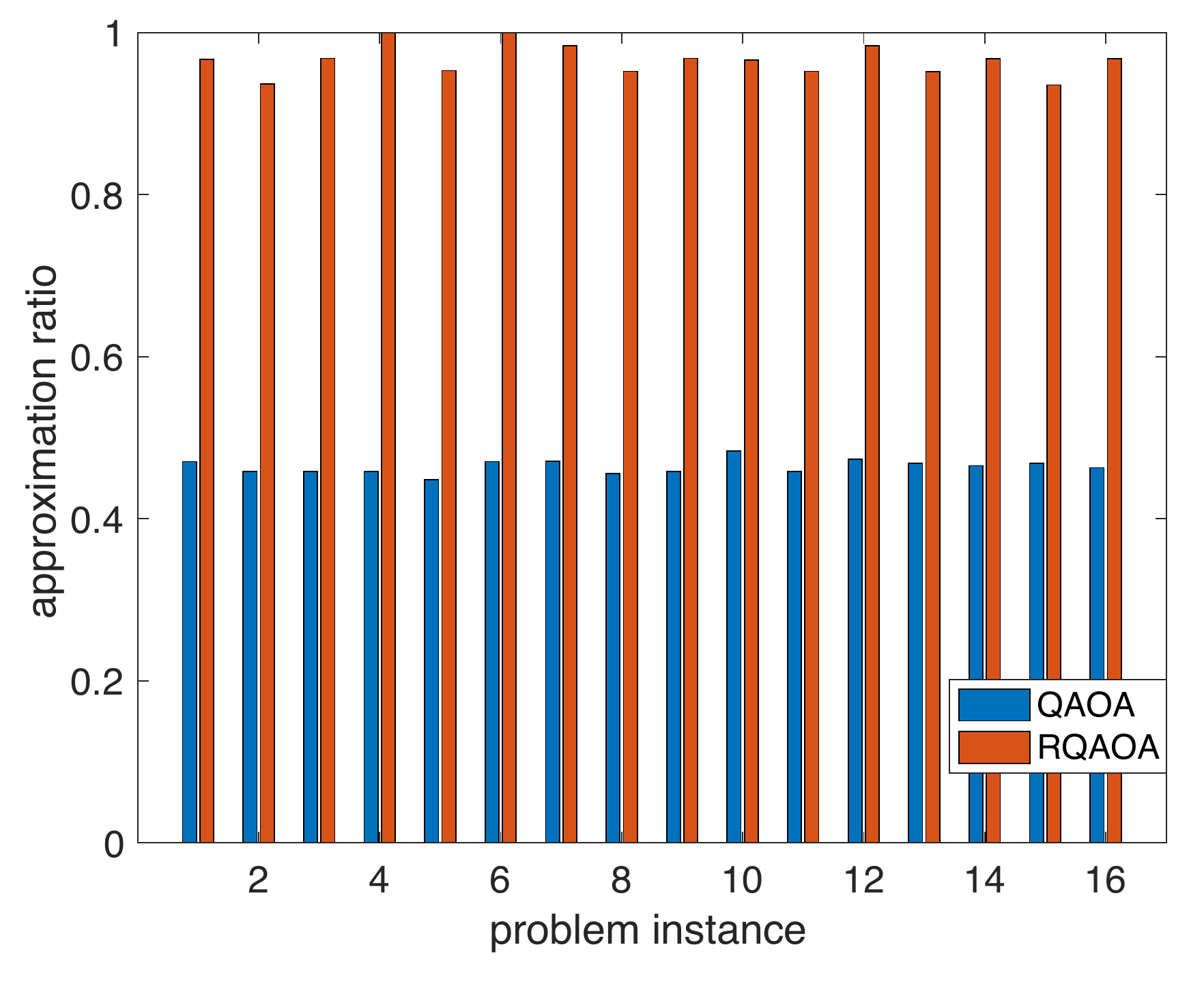} }
\caption{\label{fig:QAOA}
Approximation ratios achieved by the level-1 QAOA (blue) and  RQAOA (red)
for the Ising-type Hamiltonian $H_n$ defined Eq.~\eqref{Ising1}
with $n=32$ (left) and $n=100$ (right).
We consider  $16$ problem instances with random $3$-regular graphs and random couplings $J_{p,q}=\pm 1$. 
The cutoff value for variable elimination was chosen as $n_c=8$ (left) and $n_c=30$ (right).
The approximation ratio achieved by 
a given algorithm is defined as $\la z|H_n|z\ra/E_{\mathrm{max}}$, where 
$z\in \{1,-1\}^n$ is the algorithm's output and $E_{\mathrm{max}}=\max_z\la z|H_n|z\ra$.
}
\end{figure*}

We compare the performance of the standard QAOA, RQAOA, and local classical algorithms by 
considering the Ising Hamiltonians in Eq.~\eqref{Ising1} with couplings $J_{p,q}=\pm 1$ defined on the cycle graph. In Appendix~D we prove: 
\begin{theorem} \label{thm:classicalvsqaoa}
For each integer $n$ divisible by $6$ there is
a family of $2^{n/3}$ Ising Hamiltonians of the form 
$H_n=\sum_{k\in\mathbb{Z}_n} J_k Z_k Z_{k+1}$ with $J_k\in \{1,-1\}$ 
such that the following holds for all Hamiltonians in the family:\\
(i)  There is a local classical algorithm which achieves the approximation ratio~$1$.\\
(ii)  Level-$p$ QAOA achieves an approximation ratio of at most $p/(p+1)$.\\
(iii) Level-$1$ RQAOA  achieves the approximation ratio~$1$.
\end{theorem}
Our definition of local classical algorithms follows~\cite{hastings2019}. We also show that the level-$1$ RQAOA achieves the optimal approximation ratio for any 1D~Ising model with coupling coefficients~$J_k\in \{1,-1\}$. 

Finally, we numerically compare level-$1$ versions of QAOA and RQAOA. We consider  Ising-type
Hamiltonians Eq.~\eqref{Ising1} with random couplings $J_{p,q}=\pm 1$ defined
on random $3$-regular graphs. For each problem instance we compute the optimal
approximation ratios achieved by the two algorithms, see Figure~\ref{fig:QAOA}. The level-$1$ RQAOA significantly outperforms the standard level-$1$ QAOA for these problem instances. The details of this simulation and further discussion of RQAOA can be found in Appendix~C.4.

\noindent
{\em Acknowledgments.}
The authors are grateful to Giacomo Nannicini and Kristan Temme for helpful discussions.
SB acknowledges the support of the IBM Research Frontiers Institute and funding from the MIT-IBM Watson AI Lab under the project Machine Learning in Hilbert space. ET acknowledges the support of the Natural Sciences and Engineering Research Council of Canada (NSERC) and funding provided by the Institute for Quantum Information and Matter, an NSF Physics Frontiers Center (NSF Grant No. PHY- 1733907). RK is supported by the Technical University of Munich -- Institute for Advanced Study, funded by the German Excellence Initiative and the European Union Seventh Framework Programme under grant agreement no.~291763. RK and AK gratefully acknowledge support by the DFG cluster of excellence 2111 (Munich Center for Quantum Science and Technology).

\onecolumngrid

\appendix
\renewcommand{\thesubsection}{\Alph{section}.\arabic{subsection}}

\section{Proof of Corollary~1}
In this appendix, we give a proof of Corollary~1 in the main text.
Here and below, we will denote the expected approximation ratio achieved by the QAOA with Hamiltonian~$H$ as
\begin{align} 
\qaoap{H}{p}&=\left(\max_{\beta,\gamma\in \mathbb{R}^p} \bra{\Psi_H(\beta,\gamma)}H\ket{\Psi_H(\beta,\gamma)}\right) \cdot \left(\max_{x\in \{0,1\}^n}\bra{x}H\ket{x}\right)^{-1}\ ,
\end{align}
where 
\begin{align}
\ket{\Psi_H(\beta,\gamma)}&=U_H(\beta,\gamma)\ket{+^n}\qquad\textrm{ and }\qquad U_H(\beta, \gamma) =\prod_{m=1}^p  \left(e^{i\beta_m B}e^{i\gamma_m H}\right)\ \label{eq:qaoaunitary}
\end{align}
for $\beta, \gamma \in \mathbb{R}^p$ and where $B=\sum_{j=1}^n X_j$. Let us first record a few general features of the QAOA for later use.

Let $G=(V,E)$ be a graph, 
$n=|V|$, $m=|E|$, 
and let $J=(J_e)_{e\in E}\in\mathbb{R}^E$ be an assignment of edge weights on $G$. Let us define the Hamiltonian $H_G(J)$ as
\begin{align}
H_G(J) = \sum_{\{u,v\}\in E}J_{\{u,v\}}Z_uZ_v\ .\label{eq:hgj}
\end{align}
It will be useful for later to also define
\begin{align}\label{eq:hgdefs}
H_G  = \sum_{\{u,v\}\in E}Z_uZ_v,\qquad\text{and}\qquad H^{\mathrm{MaxCut}}_G = \frac{1}{2}(mI-H_G),
\end{align}
where $H^{\mathrm{MaxCut}}_G$ is the Hamiltonian used in QAOA for the Maximum Cut problem on the graph $G$. We will use the following bound on the circuit depth of a QAOA unitary.
\begin{lemma}\label{lem:qaoablow-up}
Let $U=U_H(\beta,\gamma)$ with $\beta, \gamma \in \mathbb{R}^p$ be a level-$p$ QAOA unitary (cf.~Eq.~\eqref{eq:qaoaunitary}) for a Hamiltonian $H=H_G(J)$  on  a graph $G$ (cf.~\eqref{eq:hgj}). Let $D$ be the maximum degree of $G$. Then $U$ can be realized by a circuit
of depth $d\leq p(D+1)$ consisting of $2$-qubit gates. 

If $G$ is $D$-regular and bipartite, then the circuit depth of $U$ can be bounded by $d\leq pD$. 
\end{lemma}
\begin{proof}
By Vizing's theorem~\cite{vizing1964} there is an edge coloring of~$G$ with at most $D+1$ colors. Taking such a coloring~$E=E_1\cup\cdots\cup E_{D+1}$, we may apply each level $e^{i\beta B}e^{i\gamma H}$ of $U$ in depth~$D+1$ by applying $(\prod_{v\in V}e^{i\beta X_v})\prod_{c=1}^{D+1} V_c(\gamma)$, where each $V_c(\gamma)=\left(\prod_{\{u,v\}\in E_c}e^{i\gamma J_{\{u,v\}}Z_uZ_v}\right)$
is a depth-$1$-circuit of two-local gates.

If $G$ is $D$-regular and bipartite, we may reduce the chromatic number upper bound from $D+1$ to $D$ since all bipartite graphs are $D$-edge-colorable by K\H{o}nig's line coloring theorem.
\end{proof}

The expected QAOA approximation ratios of suitably related instances are identical:
\begin{lemma}\label{lem:covarianceqaoa}
Let $\cL\subset V$ be an arbitrary subset of vertices
and $\partial \calL$ be the set of edges that have exactly one endpoint in $\calL$.
Let $J=(J_e)_{e\in E}\in\mathbb{R}^{E}$ be arbitrary edge weights. Define $\tilde{J}=(\tilde{J}_e)_{e\in E}\in \mathbb{R}^E$ by 
\begin{align}
\tilde{J}_e&=
\begin{cases}
-J_e\qquad&\textrm{ if }e\in \partial \cL\\
J_e\qquad &\textrm{ otherwise}\ .
\end{cases}
\end{align}
Then expected QAOA approximation ratios satisfy\begin{align}
\qaoap{H_G(J)}{p}&=\qaoap{H_G(\tilde{J})}{p}.
\end{align}
\end{lemma}
\begin{proof}
Let us write $H=H_G(J)$ and $\tilde{H}=H_G(\tilde{J})$ for brevity. Let $\overline{X}=\overline{X}[\cL]$ be a tensor product of Pauli-$X$ operators acting on every qubit in $\cL\subset V$. Then $\tilde{H}=\overline{X} H\overline{X}$, which implies that
\begin{align}
\max_{x\in\{0,1\}^n} \bra{x}H\ket{x}&=\max_{x\in \{0,1\}^n}\bra{x}\tilde{H}\ket{x}\ .\label{eq:optimumconditionxxtilde}
\end{align}

Let  $\beta, \gamma \in \mathbb{R}^p$ be arbitrary. Then we also have
\begin{align}
 \overline{X}
\ket{\Psi_{\tilde{H}}(\beta,\gamma)}&=\prod_{m=1}^p (\overline{X}e^{i\beta_m B}e^{i\gamma_m \tilde{H}}\overline{X}) \ket{+^n} = \prod_{m=1}^p (e^{i\beta_m B}e^{i\gamma_m H}) \ket{+^n} = |\Psi_H(\beta,\gamma)\rangle,
\end{align}
where identities in the middle follow since $|+^n\rangle$ is stabilized by $\overline{X}$, and since $[\overline{X},B]=0$. Therefore we have
\begin{align}
\bra{\Psi_{\tilde{H}}(\beta,\gamma)}\tilde{H}
\ket{\Psi_{\tilde{H}}(\beta,\gamma)}=
\bra{\Psi_{\tilde{H}}(\beta,\gamma)}\overline{X}  H \overline{X}
\ket{\Psi_{\tilde{H}}(\beta,\gamma)} = \bra{\Psi_{H}(\beta,\gamma)} H\ket{\Psi_{H}(\beta,\gamma)}.
\end{align}
Combined with~\eqref{eq:optimumconditionxxtilde}, this implies the claim.
\end{proof}

In particular, if $G=(V,E)$ is a bipartite graph, then Lemma~\ref{lem:covarianceqaoa} implies that
\begin{align}
\qaoap{H_G}{p} = \qaoap{-H_G}{p}\ 
\end{align}and 
\begin{align}
\qaoap{H_G^{\mathrm{MaxCut}}}{p} = \frac{1}{2}(1 + \qaoap{H_G}{p})\label{eq:hmaxcutrelationx}.
\end{align}

We now prove Corollary~1. It is a direct consequence of Theorem~1, which we restate here for convenience in the notation of this appendix:\setcounter{theorem}{0}
\begin{theorem}
Consider a family $\{G_n=([n],E_n)\}_{n\in\cI}$ of graphs with Cheeger constant lower bounded 
as $h(G_n)\geq h>0$ for all $n\in\cI$.  Then
\begin{align}
\bra{\varphi}U^\dagger H_{G_n} U\ket{\varphi}< |E_n|-\frac{hn}{3}
\end{align}
for any $\mathbb{Z}_2$-symmetric depth-$d$ circuit $U$ composed of two-qubit gates, any $\mathbb{Z}_2$-symmetric product state $\varphi$, and any $n> 48^2 8^d$, $n\in\cI$. 
\end{theorem}

Then we have the following:\setcounter{theorem}{0}
\begin{corollary}
For every integer $D\ge 3$ there exists an
infinite family  of bipartite $D$-regular graphs $\{G_n\}_{n\in \cI}$
such that
\begin{align}
\qaoap{H_{G_n}^{\mathrm{MaxCut}}}{p}\le \frac{5}{6} + \frac{\sqrt{D-1}}{3D}
\end{align}
as long as 
\be
p<(1/3\, \log_2n - 4)D^{-1}\ .\label{eq:pconditionlowerbound}
\ee
\label{corol:QAOA1_app}
\end{corollary}
\begin{proof}
Fix some $D\ge 3$. By the results of~\cite{marcus2015a,marcus2015b}, there exists 
an infinite family $\{G_n\}_{n\in\mathcal{I}}$ of bipartite $D$-regular Ramanujan graph with~$n$ vertices for every $n\in\mathcal{I}$. Consider a fixed $n\in\mathcal{I}$ and let $p=p(n)$ be the associated QAOA level. 
Let $U_n=U_{H_{G_n}}(\beta^*,\gamma^*)$ be a level-$p$ QAOA unitary for the Hamiltonian~$H_{G_n}$ on $G_n$,
and assume that $\beta^*,\gamma^*\in \mathbb{R}^p$  are such that the expectation of~$H_{G_n}$ is maximized. Because~$G_n$ is $D$-regular, 
the circuit depth of  $U_n$ can be bounded from above by $pD$ according to Lemma~\ref{lem:qaoablow-up}. Condition~\eqref{eq:pconditionlowerbound} implies that 
$n>48^2 8^{pD}$, thus
\begin{align}
\qaoap{H_{G_n}}{p}&=\frac{1}{|E_n|} \bra{+^n} U_n^\dagger H_{G_n} U_n\ket{+^n}<1-\frac{h}{3|E_n|}n =1-\frac{2h}{3D}\ 
\end{align}
by Theorem~1, where we have used that $|E_n|=n D/2$.
With~\eqref{eq:hmaxcutrelationx} (using that $G_n$ is bipartite)
we conclude that
\begin{align}
\qaoap{H_{G_n}^{\mathrm{MaxCut}}}{p}&<1-\frac{h}{3D} \ .
\end{align}
The claim then follows from the bound
$h/D \ge (D-2\sqrt{D-1})/(2D)$, valid for all Ramanujan graphs. 
\end{proof}

\section{Optimal variational circuit for the ring of disagrees}
\label{app:tight}

In this section we prove that the upper bound of Theorem~2 in the main text
is tight whenever $n$ is a multiple of $2R+1$.
Let 
\[
|\ghz{n}\ra = 2^{-1/2}(|0^n\ra + |1^n\ra)
\]
be the GHZ state of $n$ qubits.
\begin{lemma}
\label{lemma:ghz1}
Suppose $n=2p+1$ for some integer $p$. There exists a $\ZZ_2$-symmetric  range-$p$ 
quantum circuit $V$ such that 
\be
\label{eq:GHZoptimal}
|\ghz{n}\ra = V|+^n\ra.
\ee
\end{lemma}
\begin{proof}
We shall write $\cx_{c,t}$ for the CNOT gate with a control qubit $c$
and a target qubit $t$.
Let $c=p+1$ be the central qubit.
One can easily check that 
\[
|\ghz{n}\ra = \left( \prod_{j=1}^p \cx_{c,c-j} \cx_{c,c+j} \right) H_c |0^n\ra.
\]
All $\cx$ gates in the product pairwise commute, so the order does not matter.
Inserting a pair of Hadamards on every qubit $j\in [n]\setminus \{c\}$
before and after the respective $\cx$ gate
and using the identity $(I\otimes H)\cx(I\otimes H)=\cz$ one gets 
\be
\label{eq:ghz2}
|\ghz{n}\ra = \left( \prod_{j\in [n]\setminus \{c\}} H_j \right) \left( \prod_{j=1}^p \cz_{c,c-j} \cz_{c,c+j} \right)  |+^n\ra.
\ee
Let $S=\exp{[i(\pi/4)Z]}$ be the phase-shift gate. 
Define the two-qubit Clifford gate
\[
\rz=(S\otimes S)^{-1} \cz = \exp(-i \pi/4)\exp{[-i(\pi/4) (Z\otimes Z)]}.
\]
Expressing $\cz$ in terms of $\rz$ and $S$ in Eq.~\eqref{eq:ghz2} one gets
\be
\label{eq:ghz3}
|\ghz{n}\ra = S_c^{2p} \left( \prod_{j\in [n]\setminus \{c\}} H_jS_j \right) \left( \prod_{j=1}^p \rz_{c,c-j} \rz_{c,c+j} \right)  |+^n\ra.
\ee
Multiply both sides of Eq.~\eqref{eq:ghz3} on the left by a product of $S$ gates over
qubits $j\in [n]\setminus \{c\}$. Noting that 
\[
SHS = i \exp{[-i(\pi/4)X]}
\]
one gets (ignoring an overall phase factor)
\be
\label{eq:ghz4}
\prod_{j\in [n]\setminus \{c\}} S_j |\ghz{n}\ra =S_c^{2p}
\left( \prod_{j\in [n]\setminus \{c\}} \exp{[-i(\pi/4)X_j]} \right) 
\left( \prod_{j=1}^p \rz_{c,c-j} \rz_{c,c+j} \right)  |+^n\ra.
\ee
Using the identity
\[
\prod_{j\in [n]\setminus \{c\}} S_j |\ghz{n}\ra=
S_c^{2p} |\ghz{n}\ra.
\] 
one can cancel $S_c^{2p}$ that appears in both sides  of Eq.~\eqref{eq:ghz4}. We arrive at
Eq.~\eqref{eq:GHZoptimal} with
\be
\label{eq:ghz5}
V = \left( \prod_{j\in [n]\setminus \{c\}} \exp{[-i(\pi/4)X_j]} \right) 
\left( \prod_{j=1}^p \rz_{c,c-j} \rz_{c,c+j} \right) 
\ee
Obviously, $V$ is $\ZZ_2$-symmetric since any individual gate commutes with $X^{\otimes n}$.
Let us check that $V$ has range-$p$. Consider any single-qubit observable $O_q$ 
acting on the $q$-th qubit. 
Consider three cases. {\em Case~1:} $q=c$. Then $V^\dag O_q V$ may be supported on all $n$ qubits.
However, $[c-p,c+p]=[1,n]$, so the $p$-range condition is satisfied trivially. 
{\em Case~2:} $1\le q<c$. Then all gates $\rz_{c,c+j}$ in $V$  cancel the corresponding gates in $V^\dag$,
so that $V^\dag O_q V$ has support in the interval $[1,c]\subseteq [q-p,q+p]$.
Thus the $p$-range condition is satisfied.
{\em Case~3:} $c<q\le n$. This case is equivalent to Case~2 by symmetry. 
\end{proof}
Recall that we consider the ring of disagrees Hamiltonian
\[
H_n = \frac12 \sum_{p\in \ZZ_n} (I-Z_p Z_{p+1}).
\]
\begin{lemma}
\label{lemma:ghz2}
Consider any integers $n,p$ such that $n$ is even and $n$ is a multiple of $2p+1$.
Then there exists a $\ZZ_2$-symmetric range-$p$ circuit $U$ such that 
\be
\la +^n|U^\dag H_n U|+^n\ra = \frac{2p+1/2}{2p+1}.
\ee
\end{lemma}
\begin{proof}
Let $V$ be the $\ZZ_2$-symmetric range-$p$ unitary operator preparing
the GHZ state on $2p+1$ qubits starting from $|+^{2p+1}\ra$, see Lemma~\ref{lemma:ghz1}.
Suppose $n=m(2p+1)$ for some integer $m$. 
Define
\[
U=U_1 U_2,  
\]
where
\[
U_1 = (X\otimes I)^{\otimes n/2} \quad \mbox{and} \quad U_2 = V^{\otimes m}.
\]
Since each copy of $V$ acts on a consecutive interval of  qubits and has range $p$, one infers that
$U$ has range $p$. 
We have 
\[
U_1^\dag H_n U_1 =  \sum_{p\in \ZZ_n} G_p,
\quad G_p =  \frac12 (I+Z_p Z_{p+1}).
\]
The state $U_2|+^n\ra$ is a tensor product of GHZ states supported on consecutive
tuples of $2p+1$ qubits. The expected value of $G_p$ on the state
$U_2|+^n\ra$  equals $1$ if $G_p$ is supported on one of the GHZ state.
Otherwise, if $G_p$ crosses the boundary between two GHZ states, 
the expected value of $G_p$ on the state $U_2|+^n\ra$  equals $1/2$.
Thus 
\[
\la +^n |U^\dag H_n U|+^n\ra = \sum_{p\in \ZZ_n} \la +^n|U_2^\dag G_p U_2|+^n\ra
=m (2p +1/2)  = n\frac{2p+1/2}{2p+1}.
\] 
\end{proof}

\section{Recursive QAOA}
\label{app:RQAOA}

In this appendix, we outline the Recursive QAOA algorithm (RQAOA) for general cost functions.

\subsection{Variable Elimination}\label{sec:variableelimination}
Let $G=(V,E)$ be a hypergraph  with $|V|=n$~vertices.
Suppose a variable $x_v\in \{1,-1\}$ is associated with each vertex $v\in V$.
Let $\{1,-1\}^V=\{1,-1\}^n$ be the set of all possible variable assignments.
Let $J:E\rightarrow\mathbb{R}$ be a function which assigns a real weight 
$J_e$ to every hyperedge~$e$. 
Given a subset  $f\subset V$ and an assignment
$x\in \{1,-1\}^V$,   let us write
\[
x(f)=\prod_{v\in f} x_v.
\] 
Let us agree that $x(\emptyset)=1$.
We consider the problem of maximizing cost functions of the form
\begin{align}
C(x)&=\sum_{e\in E}J_e x(e)\ 
\end{align}
over $x\in \{1,-1\}^V$. 

Fix some vertex $v\in V$. 
As a motivation,  we first describe how a single variable $x_v$
can be eliminated when a suitably constrained problem is considered. 
Namely, suppose that  instead of trying to approximate $\max_{x\in \{1,-1\}^V} C(x)$, we restrict to 
$x\in \{1,-1\}^V$ satisfying 
\begin{align}
x(f)&=\sigma,\label{eq:zeliminationconstraint}
\end{align}
where  $f\subset V$ is some fixed subset of vertices containing~$v$, and $\sigma\in \{1,-1\}$ is a constant.
If $x\in \{1,-1\}^V$ satisfies the constraint~\eqref{eq:zeliminationconstraint}  then
\[
J_e x(e)=J_e x(e) x(f) \sigma = \sigma J_e x(e  \bigtriangleup f).
\]
Here and below $A\bigtriangleup B$ denotes the symmetric difference of sets $A$ and $B$.
We arrive at
\begin{align}
C(x)&=\sum_{\substack{e\in E\\
v\not\in e}}J_e x(e)+\sum_{\substack{e\in E:\\
v\in e}} \sigma J_e x\left(e \bigtriangleup f\right)\ . \label{eq:czmodified}
\end{align}
Note that $C(x)$ does not depend on $x_v$. Expression~\eqref{eq:czmodified} can be written as a sum over the hyperedges of a  hypergraph~$G'=(V',E')$ with vertex set $V'=V\backslash \{v\}$ and hyperedges 
\be
\label{E'1}
E'=E'_0 \cup E'_1,
\ee
where
\be
\label{E'2}
E'_0=\{ e\in E\, : \, v\notin e\} \quad \mbox{and} \quad E'_1 = \{ e \bigtriangleup f\, : \,  e\in E, \quad v\in e\}.
\ee
Note that $G'$ no longer contains the vertex~$v$. Define a function $J':E'\rightarrow\mathbb{R}$ such that
\be
\label{J'1}
J'_e = J_e \quad \mbox{if $e\in E'_0$},
\ee
and
\be
\label{J'2}
J'_e =  \sigma J_{e  \bigtriangleup f}  \quad \mbox{if $e\in E'_1$}.
\ee
By construction, 
the maximum of
\begin{align}
C'(x)&=\sum_{e\in E'}J_e' x(e)\label{eq:derivedproblem}
\end{align}
over all assignments $x\in \{1,-1\}^{V'}$
coincides with the maximum of
$C(x)$ over all $x\in \{1,-1\}^V$ satisfying the constraint Eq.~\eqref{eq:zeliminationconstraint}. Furthermore,  
any maximum~$x^*$ of~$C'$ can directly be translated to a corresponding maximum of~$C$  
 over the restricted set defined by the constraint~\eqref{eq:zeliminationconstraint} by setting $x^*_v=\sigma\cdot x(f\backslash \{v\})$. That is, we have  
 $x^*=\xi(x)$ for the function~$\xi:\{1,-1\}^{V'}\rightarrow\{1,-1\}^V$ defined by 
 \begin{align}
 \xi(x)_w&=\begin{cases}
 \sigma\cdot x(f\backslash \{v\})\qquad&\textrm{ for }v=w\\
 x_w\qquad &\textrm{ otherwise }  
 \end{cases} \label{eq:xidef}
 \end{align}
 for all $w\in V$. 
 
 In summary, we have  reduced the problem of maximizing $C(x)$ over $n$ variables 
 $x\in \{1,-1\}^V$ satisfying~\eqref{eq:zeliminationconstraint} to the problem of maximizing $C'(x)$ over
 $n-1$ variables  $x\in \{1,-1\}^{V'}$. If a global  maximum~$x$ of $C(x)$ happens to satisfy~\eqref{eq:zeliminationconstraint}, the new reduced problem yields a solution to the original problem.

\subsection{Correlation rounding\label{sec:corrrounding}}
To construct an approximation algorithm, we simply {\em impose} a  constraint of the form~\eqref{eq:zeliminationconstraint} by choosing $f\subset V$, $v\in f$ and $\sigma\in \{1,-1\}$ appropriately. To make the latter choice, we use the standard $\QAOA_p$ algorithm with $p=O(1)$. That is, let us set
\begin{align}
H_G(J)=\sum_{e\in E} J_e Z(e)\qquad\textrm{ where }\qquad Z(e)=\prod_{v\in e}Z_v\ 
\label{H_G(J)}
\end{align}
and  write $H=H_G(J)$. We first use the standard $\QAOA_p(H)$ algorithm to find an optimal state 
\[
\Psi=\Psi_{H_G}(\beta_*,\gamma_*)\in (\mathbb{C}^2)^{\otimes |V|}
\] 
maximing the energy of~$H_G$. 
The expected value 
\[
M_e = \bra{\Psi}Z(e)\ket{\Psi}
\]
can be efficiently approximated on a quantum computer for any $e\in E$.

Suppose the state~$\Psi$ is measured in the computational basis giving a string $x\in \{1,-1\}^V$. Clearly, if $|M_f|$ is close to~$1$, then the variables $\{x_v\}_{v\in f}$  satisfy a constraint of the form~\eqref{eq:zeliminationconstraint}  with high probability with $\sigma=\mathsf{sign}(M_f)\in \{1,-1\}$ and any $v\in f$. Thus it is natural to choose $f$ such that $|M_f|$ is maximal.  Combined  with the procedure for eliminating the corresponding variable~$x_v$ described in Section~\ref{sec:variableelimination}, we obtain a subroutine
for reducing the  problem size by one variable. Pseudocode for this routine is given  below.

Imposing a constraint of the form~\eqref{eq:zeliminationconstraint} 
 can be viewed as rounding correlations
among the variables $\{x_w\}_{w\in f}$: indeed, the constraint demands that for $v\in f$, the variable $x_v$ and $x(f\backslash \{v\})$  must be perfectly correlated
or anti-correlated.

  \begin{figure}[ht!]
  \begin{algorithmic}[1]
  \Statex
  \Function{eliminateVariable}{$G=(V,E),J$}\\ \hrulefill\\
  \textbf{Input:} A hypergraph $G=(V,E)$ and a weight function $J:E\rightarrow\mathbb{R}$\\
  \textbf{Output:} A hypergraph $G'=(V',E')$, 
   $J':E'\rightarrow\mathbb{R}$ and   a function  $\xi:\{1,-1\}^{V'}\rightarrow \{1,-1\}^V$.     \\ \hrulefill
   \State{Run  $\QAOA_p(H_G(J))$ to find a state~$\Psi$ which maximizes $\bra{\Psi}H_G(J)\ket{\Psi}$.}
   \State{Compute $M_e=\bra{\Psi}Z(e)\ket{\Psi}$ for every $e\in E$.}
   \State{Set $f=\argmax_{f\in E} |M_f|$ (breaking ties arbitrarily).}
     \State{Pick $v\in f$ arbitrarily.}
    \State{Set $\sigma=\mathsf{sign}(M_{f})$.}
  \State{Define $V'=V\backslash\{v\}$. Also define $E'$, $J'$ 
  and $\xi$  by Eqs.~\eqref{E'1},\eqref{E'2},\eqref{J'1},\eqref{J'2},\eqref{eq:xidef}.}
     \State{\Return $(G'=(V',E'), 
     J',\xi)$.}
   \EndFunction
 \end{algorithmic}
\end{figure}

\subsection{The recursive QAOA (\texorpdfstring{$\RQAOA$}{RQAOA}) algorithm\label{sec:recursiveQAOAexplained}}
The recursive QAOA algorithm ($\RQAOA$) we propose here proceeds simply by
iterating the process of eliminating one variable at a time until the number of variables reaches some specified threshold value~$n_c\ll n$. The remaining instance of the problem with~$n_c$ variables is solved by a purely classical algorithm (for example, by the brute force method). Thus the value of~$n_c$ controls  how the workload is distributed between the quantum and the classical computers. Pseudocode for the $\RQAOA$ algorithm is given in Fig.~\ref{fig:pseudocodeRQAOA}.  

  \begin{figure}[ht]
  \begin{algorithmic}[1]
  \Statex
  \Function{$\RQAOA$}{$G=(V,E),J$}\\ \hrulefill\\
  \textbf{Input:} A hypergraph $G=(V,E)$ with $n=|V|$ and a weight function $J:E\rightarrow\mathbb{R}$ defining a 
  Hamiltonian $H_G(J)$, see Eq.~\eqref{H_G(J)}.\\
  \textbf{Output:} A variable assignment  $x\in \{-1,1\}^V$\\ \hrulefill
   \State{Let $\xi^{(0)} \, : \, \{1,-1\}^V\rightarrow\{1,-1\}^V$ be the identity map.}
            \For{$k=1$ to $n-n_c$}
            \State{$(G, 
            J,\xi)\gets \mathsf{eliminateVARIABLE}(G,J)$.}
          \State{$\xi^{(k)}\gets \xi^{(k-1)}\circ \xi$.}
                   \EndFor
                   \State{Let $G=(V,E)$ be the final hypergraph with $|V|=n_c$ vertices.}
                     \State{Find $x^*=\argmax_{x\in \{1,-1\}^{V}}
                     \bra{x}H_{G}(J)\ket{x}$.}
                     
       \State{\Return $\xi(x^*)$}
   \EndFunction
 \end{algorithmic}
\caption{Pseudocode for the recursive QAOA algorithm. \label{fig:pseudocodeRQAOA}}
\end{figure}

\subsection{Classical simulability of level-$1$ \texorpdfstring{$\RQAOA$}{RQAOA} for Ising models\label{sec:classicalsimulability}}

Suppose $J$ is a real symmetric matrix of size $n$.
Here we consider Ising-like cost functions such that the corresponding Hamiltonian is 
\begin{align}
H&=\sum_{1\leq p<q\leq n} J_{p,q} Z_pZ_q\ .\label{eq:isingcostfct}
\end{align}
The mean values of a Pauli operator $Z_p Z_q$ on the level-$1$ QAOA
state  
\[
|\Psi_H(\beta,\gamma)\ra =e^{i\beta B} e^{i\gamma H} |+^n\ra
\]
can be computed in time $O(n)$ using an an explicit analytic formula. 
Such a formula was derived for the Max-Cut cost function by Wang et al.~\cite[Theorem 1]{zhihui2018}. 
Below we provide a generalization to general Ising
Hamiltonians.
Since the total number of terms in the cost function is~$O(n^2)$, simulating each step of $\RQAOA$
takes time at most $O(n^3)$. Assuming that $n_c = O(1)$, the number of steps is roughly $n$ so
that the full simulation cost is $O(n^4)$. Crucially, the simulation cost of this method does
not depend on the depth of the variational circuit. This is important because $\RQAOA$ may
potentially increase the depth from $O(1)$ to $O(n)$ since it adds many new terms to the cost
function.

\begin{lemma}
\label{lemma:QAOAlevel1}
Fix a pair of qubits $1\le u<v\le n$. 
Let $c=\cos{(2\beta)}$ and $s= \sin{(2\beta)}$. Then
\begin{eqnarray}
 \la \Psi_H(\beta,1)| Z_u Z_v |\Psi_H(\beta,1)\ra
&=& (s^2/2) \prod_{p\ne u,v} \cos{[2J_{u,p}   - 2J_{v,p} ]}
- (s^2/2)  \prod_{p\ne u,v} \cos{[2J_{u,p}   + 2J_{v,p} ]} \nonumber \\
&& + cs \cdot \sin{(2J_{u,v})}\left[ 
\prod_{p\ne u,v} \cos{(2J_{u,p})} + \prod_{p\ne u,v} \cos{(2J_{v,p})}\right].
\label{ZZmain}
\end{eqnarray}
\end{lemma}
Here we only consider the case $\gamma=1$ since 
$\gamma$ can be absorbed into the definition of $J$.
\begin{proof}
Given a 2-qubit observable $O$ define the mean value 
\be
\mu(O)= \la \Psi_H(\beta,1)| O_{u,v} |\Psi_H(\beta,1)\ra.
\ee
We are interested in the observable  $O=ZZ\equiv Z\otimes Z$.

We note that all terms in $H$ and $B$ that act trivially on $\{u,v\}$ do not contribute
to $\mu(O)$. Such terms can be set to zero. 
Given a 2-qubit observable $O$, define  a mean value
\be
\label{mu'(O)}
\mu'(O)= \la +^n|e^{iH'} O_{u,v}  e^{-iH'}|+^n\ra,  \quad \mbox{where} \quad 
\quad H'=\sum_{p\ne u,v} (J_{u,p} Z_u  + J_{v,p} Z_v )Z_p.
\ee
Using the identities
\begin{eqnarray}
e^{i\beta (X_u+X_v)} Z_u Z_v e^{-i\beta (X_u+X_v)} &=& c^2 Z_u Z_v
+s^2 Y_u Y_v + cs(Z_u Y_v + Y_u Z_v), \nonumber \\
e^{iJ_{u,v} Z_u Z_v} Z_u Z_v e^{-iJ_{u,v} Z_u Z_v}  &=& Z_u Z_v,\nonumber \\
e^{iJ_{u,v} Z_u Z_v} Y_u Y_v e^{-iJ_{u,v} Z_u Z_v}  &=& Y_u Y_v\nonumber \\
e^{iJ_{u,v} Z_u Z_v} Z_u Y_v e^{-iJ_{u,v} Z_u Z_v} 
&=& \cos{(2J_{u,v})} Z_u Y_v + \sin{(2J_{u,v})} X_v, \nonumber \\
e^{iJ_{u,v} Z_u Z_v} Y_u Z_v e^{-iJ_{u,v} Z_u Z_v} 
&=& \cos{(2J_{u,v})} Y_u Z_v + \sin{(2J_{u,v})} X_u,
\end{eqnarray}
and noting that $\mu'(ZZ)=0$  one easily gets
\be
\label{ZZ1}
\mu(ZZ) = s^2 \cdot \mu'(YY) + cs \cdot \cos{(2J_{u,v})}\left[ \mu'(ZY) + \mu'(YZ)\right]
+ cs \cdot \sin{(2J_{u,v})} \left[ \mu'(XI) + \mu'(IX)\right].
\ee
Using the explicit form of $H'$ one gets
\be
\label{state1}
e^{-i H'}|+^n\ra = \frac12 \sum_{a,b=0,1}   |a,b\ra_{u,v} \otimes |\Phi(a,b)\ra_{\mathsf{else}},
\ee
where $|\Phi(a,b)\ra$ is a tensor product state of $n-2$ qubits defined by 
\be
|\Phi(a,b)\ra = \bigotimes_{p\ne u,v} |J_{u,p} (-1)^a + J_{v,p}(-1)^b\ra_p
\quad \mbox{where} \quad 
|\theta\ra\equiv e^{-i\theta Z}|+\ra.
\ee
Combining Eqs.~\eqref{mu'(O)},\eqref{state1} one gets
\be
\label{mu'}
\mu'(O)=(1/4) \sum_{a,b,a',b'=0,1} \la a',b'|O|a,b\ra \cdot \la \Phi(a',b')|\Phi(a,b)\ra.
\ee
Using the tensor product form of the states $|\Phi(a,b)\ra$ 
and the identity $\la \theta'|\theta\ra = \cos(\theta-\theta')$
gives
\be
\label{inner}
\la \Phi(a',b')|\Phi(a,b)\ra = \prod_{p\ne u,v} \cos{[J_{u,p} (-1)^a   -J_{u,p} (-1)^{a'} + J_{v,p} (-1)^b
- J_{v,p} (-1)^{b'}]}.
\ee
From Eqs.~\eqref{mu'},\eqref{inner} one can easily compute the mean value $\mu'(O)$
for any 2-qubit observable.

Consider first the case $O=YY$.
Then the only terms contributing to Eq.~\eqref{mu'}
are those with $a'=a\oplus 1$ and $b'=b\oplus 1$. 
The identity $\la a\oplus 1|Y|a\ra = -i(-1)^a$ gives
\be
\mu'(YY)=-(1/4)\sum_{a,b=0,1} (-1)^{a+b} 
 \prod_{p\ne u,v} \cos{[2J_{u,p} (-1)^a   + 2J_{v,p} (-1)^b]},
\ee
that is,
\be
\label{YY}
\mu'(YY)=(1/2)\prod_{p\ne u,v} \cos{[2J_{u,p}   - 2J_{v,p} ]}
-(1/2)\prod_{p\ne u,v} \cos{[2J_{u,p}   + 2J_{v,p} ]}.
\ee
Next consider the case $O=YZ$. 
Note that the matrix elements $\la a',b'|O|a,b\ra$ have zero
real part. From Eqs.~\eqref{mu'},\eqref{inner} one infers
that $\mu'(YZ)$ has zero real part. This implies
\be
\label{YZ}
\mu'(YZ)=\mu'(ZY)=0.
\ee
Finally, consider the case $O=XI$. 
Then the only terms that contribute to Eq.~\eqref{mu'} are those
with $a'=a\oplus 1$ and $b'=b$. We get
\be
\label{XI}
\mu'(XI) = \prod_{p\ne u,v} \cos{(2J_{u,p})}.
\ee
Here we noted that the inner product Eq.~\eqref{inner} with $a'=a\oplus 1$
and $b'=b$ does not depend on $a,b$. By the same argument,
\be
\label{IX}
\mu'(IX) = \prod_{p\ne u,v} \cos{(2J_{v,p})}.
\ee
Combining  Eq.~\eqref{ZZ1} and Eqs.~\eqref{YY},\eqref{YZ},\eqref{XI},\eqref{IX} one arrives at
Eq.~\eqref{ZZmain}.
\end{proof}

For more general cost functions including interactions among three or more variables, there are two complications: First, 
unlike in the Ising case,
the variable elimination process will typically increase the degree of non-locality of interactions.
Second, mean values of Pauli operators on the QAOA state $\Psi_H(\beta,\gamma)$ lack a simple analytic formula (as
far as we know). However, one can approximately compute the mean values using the Monte Carlo method due to Van den Nest~\cite{vandennest}. A specialization of this method to simulation of the level-1 QAOA is described in~\cite{bravyietalsimulation18}. The Monte Carlo simulator has runtime scaling polynomially with the number of qubits, number of terms in the cost function, and the inverse error tolerance, see~\cite{bravyietalsimulation18} for details. This method also requires no restrictions on the depth of the variational circuit.

An important distinction between QAOA and RQAOA
lies in the measurement step. QAOA requires few-qubit measurements to estimate the variational energy  as well as the final $n$-qubit measurement
that assigns a  value to each individual variable.
This last step is what makes QAOA hard to simulate classically and may lead to a quantum advantage~\cite{farhiharrow2016}. In contrast, RQAOA only needs few-qubit measurements to estimate mean values
of individual terms in the cost function. The $n$-qubit measurement step is replaced by the correlation rounding that eliminates variables one by one. One may ask whether the lack of multi-qubit measurements also precludes a quantum advantage. Indeed, in the special case of level-1 variational circuits and the Ising-like cost function
RQAOA can be efficiently simulated classically, see above.
However, level-$p$ RQAOA with $p>1$ as well as level-1 RQAOA with more general cost functions are not known to be classically simulable in polynomial time,  leaving room for a quantum advantage.

\section{Comparison of QAOA, RQAOA, and Classical Algorithms}

\subsection{QAOA versus Classical Local Algorithms} \label{sec:qaoalocalclassical}
In this section, we discuss another limitation of QAOA which results from its locality and the covariance condition discussed in Lemma~\ref{lem:covarianceqaoa}: we compare QAOA  to a certain very simple classical local algorithm (see Lemma~\ref{lem:classicallocal} below). We show that there is an exponential number of problem instances for which the classical local algorithm outperforms QAOA. 

Let us briefly sketch the notion of a local classical algorithm. We envision that the tuple $(J_e)_{e\in E}$ is given as input. 
Here we are interested in  algorithms which are local with respect to the underlying graph~$G$. For $r\in\mathbb{N}$ and $v\in V$, define
\begin{align}
E_r(v)=\bigcup_{\ell=1}^r \bigcup_{\substack{(e_1,\ldots,e_\ell)\\
\textrm{ path with $v\in e_1$}}} \{e_1,\ldots,e_\ell\}
\end{align}
to be the set of edges that belong to a path starting at~$v$ of length bounded by~$r$. Consider a classical  algorithm~$\cA$ which on input $\{J_e\}_{e\in E}$ outputs $x=(x_1,\ldots,x_n)\in \{0,1\}^n$. We say that~$\cA$ is $r$-local if 
there is a family of functions $\{g_v:\mathbb{R}^{E_r(v)}\rightarrow\{0,1\}\}_{v\in V}$ such that the following holds for every problem instance $(J_e)_{e\in E}\in \mathbb{R}^{E}$: We have 
\begin{align}
x_v&=g_v\left(\{J_e\}_{e\in E_r(v)}\right)\qquad\textrm{ for every }v\in V\ .\label{eq:localfunctionsclassical}
\end{align} 
In other words, in an $r$-local classical algorithm, every  output bit~$x_v$ only depends on edge weights~$J_e$ belonging to paths of length bounded by~$r$ starting at~$v$. 
We note that this definition can easily be generalized to the probabilistic case (e.g., by including local random bits). For the purposes of this section, deterministic functions turn out to be sufficient.

The (choice of) family~$\{g_v\}_{v\in V}$ can be considered as a set of variational parameters for the classical algorithm. To keep the number of variational parameters constant, we consider vertex-transitive graphs~$G$. Fix $v_*\in V$. For every $v\in V$, fix an automorphism~$\pi_v$ of~$G$ such that $\pi_v(v_*)=v$. Then the sets $E_r(v)$ for different $v\in V$ can be identified via $E_r(v)=\pi_v(E_r(v^*))$. We say that an $r$-local classical algorithm is {\em uniform} if (after this identification) $g_v\equiv g$ for all $v\in V$, i.e., if there is a single function $g:\mathbb{R}^{E_r(v_*)}\rightarrow \{0,1\}$ specifying the behavior of the algorithm.  To obtain general-purpose algorithms (applicable to any instance), the function $g:\mathbb{R}^{E_r(v_*)}\rightarrow \{0,1\}$ should be chosen adapatively (i.e., potentially depending on the instance). The definition of local classical algorithm sketched here includes e.g., the algorithms considered in Ref.~\cite{hastings2019}, though it is slightly more general as the local functions can be arbitrary.

Let $n=6r$ be a multiple of~$6$. Consider $n$-qubit Hamiltonians (cf.~\eqref{eq:hgj}) of the form 
\be
H(J)= \sum_{k\in \ZZ_n} J_kZ_k Z_{k+1}\qquad\textrm{ where }\qquad J=(J_0,\ldots,J_{n-1})\in \{1, -1\}^n\ .
\ee
To define locality and uniformity for the  cycle graph~$\mathbb{Z}_n$, let  $\pi_v(w)=v+w\pmod n$ be chosen as translation modulo~$n$
for $v\in\mathbb{Z}_n$. We show the following: 
\begin{lemma} \label{lem:classicallocal}
There is a subset~$\cS\subset \{1, -1\}^{n}$ of $2^{n/3}$ problem instances such that  the following holds:
\begin{enumerate}[(i)]
\item\label{it:qaoanogo}
$\qaoap{H(J)}{p}\leq p/(p+1)$ for every $p\in\mathbb{N}$ and every $J\in\cS$.
\item\label{it:classicalalgorithmqaoa}
There is a $1$-local uniform classical algorithm such that for every $J\in \cS$,
the algorithm outputs $x\in \{0,1\}^n$ such that $\bra{x}H(J)\ket{x}=1$.

\item\label{it:rqaoa} Level-$1$ $\RQAOA$ achieves the approximation ratio~$1$.
\end{enumerate}
\end{lemma}
\begin{proof}
For every $s=(s_0,\ldots,s_{2r-1})\in \{0,1\}^{2r}$ define $J=J(s)\in \{1, -1\}^n$ by
\begin{align}
J_{3a}=J_{3a+1}=(-1)^{s_a},\qquad\text{and}\qquad J_{3a+2}=1,
\end{align}
for all $a=0,1,\ldots,2r-1$.  We claim that the set $\cS=\{J(s)\ |\ s\in \{0,1\}^{2r}\}$ has the required properties. Consider an instance $H(J(s))$ with $s\in \cS$. Define 
\begin{align}
\overline{X}(s)&=\prod_{a=0}^{2r-1} X_{3a+1}\ .
\end{align}
Then $H(J(s))$ is related to $\HZn=\sum_{j\in \ZZ_n} Z_j Z_{j+1}$ by the gauge transformation
\begin{align}
H(J(s))&=
\overline{X}(s) \HZn
\overline{X}(s)^{-1}\ .
\end{align}
Since the QAOA algorithm is invariant under such gauge transformation (see Lemma~\ref{lem:covarianceqaoa}),
we obtain
\begin{align}
\qaoap{H(J(s))}{p}=\qaoap{\HZn}{p}\leq \frac{p}{p+1}\ 
\end{align}
where we use the bound
\begin{align}
\qaoap{\HZn^{\mathrm{MaxCut}}}{p} \leq \frac{2p+1}{2p+2}, \label{eq:provenupperbound}
\end{align}
proven in~\cite{mbengfaziosantoro2019} for even $n$, in combination with Lemma~\ref{lem:covarianceqaoa}. This shows~\eqref{it:qaoanogo}.

For the proof of~\eqref{it:classicalalgorithmqaoa}, consider the classical algorithm~$\cA$ which on input $J=(J_0,\ldots,J_{n-1})$ outputs 
\begin{align}
x_v&=g(J_{v-1},J_v)\qquad\textrm{ for every } v\in\mathbb{Z}_n\ ,
\end{align}
where
\begin{align}
g(J,J')&=\begin{cases}
1 \qquad &\textrm{ if }(J,J')=(-1,-1)\\
0&\textrm{ otherwise}\ .
\end{cases}
\end{align}
Clearly, the algorithm~$\cA$ is uniform and $1$-local, and it is easy to check that  the output satisfies~$\bra{x}H(J)\ket{x}=1$.

The proof of~\eqref{it:rqaoa} is given as a part of~Lemma~\ref{lem:rqaoa}.
\end{proof}

\subsection{\texorpdfstring{$\RQAOA$}{RQAOA} on the Ising ring\label{sec:rqaoaring}}
Here we prove that level-$1$ $\RQAOA$ achieves approximation ratio~$1$ on the ring of disagrees, in sharp contrast to (arbitrary) level-$p$ QAOA (see Lemma~\ref{lem:classicallocal}\eqref{it:qaoanogo}). More generally, level-$1$ $\RQAOA$ produces~$x\in \{0,1\}^n$ which maximizes the cost function 
(that is, achieves approximation ratio~$1$) for any $1D$~Ising model where the coupling coefficients are either~$+1$ or $-1$. 
\begin{lemma}\label{lem:rqaoa}
Consider a cost function of the form 
\begin{align}
C_n(x) &=\sum_{k\in \mathbb{Z}_n} J_k x_kx_{k+1}\qquad\textrm{ for }x\in \{1,-1\}^n\ ,
\end{align}
where $J_k\in \{1,-1\}$ for all $k\in\mathbb{Z}_n$.
Then the level-$1$ $\RQAOA$  produces $x^*\in \{1,-1\}^n$ such that $C_n(x^*)=\max_{x\in \{1,-1\}^n} C_n(x)$.
\end{lemma}

\noindent It would be interesting to see additional, more general families of examples where approximation ratios achieved by~$\RQAOA$ can be computed or lower bounded analytically.
\begin{proof}
Let 
\begin{align}
H&=\sum_{k\in\mathbb{Z}_n}J_kZ_{k}Z_{k+1}\ .\label{eq:Hznwithcpl}
\end{align}
Observe first that $\bra{\Psi_{H}(\beta,\gamma)}Z_iZ_j\ket{\Psi_{H}(\beta,\gamma)}=0$ if $|i-j|>2$ since in this case the operators $U^{-1}Z_iU$ and $U^{-1}Z_jU$ have disjoint support. 
Lemma~\ref{lemma:QAOAlevel1} shows that a $\QAOA_1$-state $\Psi_{H}(\beta,\gamma)$ has expectation values
\begin{align}
\bra{\Psi_{H}(\beta,\gamma)}Z_iZ_j\ket{\Psi_{H}(\beta,\gamma)}&=
\begin{cases}\frac{1}{2}J_{i}\sin(4\beta)\sin(4\gamma)
\qquad&\textrm{ if }j=i+1\\
\frac{1}{4}J_iJ_{i+1}\sin^2(2\beta)\sin^2(4\gamma)&\textrm{ if }j=i+2\\
0 &\textrm{ otherwise}
\end{cases}\label{eq:distijx}
\end{align}
when  $J_{k}\in \{1, -1\}$ for every $k\in \mathbb{Z}_n$. 
Thus  
\begin{align}
|\bra{\Psi(\beta,\gamma)}Z_{i}Z_{i+2}\ket{\Psi(\beta,\gamma)}|\leq 1/4\qquad\textrm{ for all  }(\beta,\gamma)\ .\label{eq:upperboundziziplustwo}
\end{align}
Assume  $(\beta^*,\gamma^*)$ are such that  $(\beta^*,\gamma^*)=\argmax_{(\beta,\gamma)}
\bra{\Psi_{H}(\beta^*,\gamma^*)}H\ket{\Psi_{H}(\beta^*,\gamma^*)}$. Then we can infer from~\eqref{eq:distijx} 
that
\begin{align}
\bra{\Psi(\beta^*,\gamma^*)}Z_{i}Z_{i+1}\ket{\Psi(\beta^*,\gamma^*)}=J_i/2\ .\label{eq:ziziplusone}
\end{align}
 Combined with~\eqref{eq:upperboundziziplustwo} and~\eqref{eq:distijx} we conclude that 
\begin{align}
\argmax_{(i,j):i<j} | \bra{\Psi_{H}(\beta^* ,\gamma^*)}Z_i Z_{j}\ket{\Psi_{H}(\beta^* ,\gamma^*)}|=(i^*,i^*+1)\label{eq:variabletoeliminate}
\end{align}
for some $i^*\in \mathbb{Z}_n$. Without loss of generality, assume that $i^*=n-2$. Then, according to~\eqref{eq:variabletoeliminate}, the $\RQAOA$ algorithm eliminates the variable~$x_{n-1}$ (i.e., $v=n-1$, $f=\{n-2,n-1\}$). By~\eqref{eq:ziziplusone}, this is achieved by imposing the constraint  
\begin{align}
x_{n-1}=x_{n-2}J_{n-2} \label{eq:znminusconstraint}
\end{align} i.e., $\sigma=J_{n-2}$. The resulting reduced graph $G'=(V',E')$ has 
vertex set $V'=V\backslash \{n-1\}=\mathbb{Z}_{n-1}$ and edges
\begin{align}
E'&=\left\{\{i,i+1\}\ |\ i\in\mathbb{Z}_n\backslash \{n-2\}\right\}\cup \{\{n-2,0\}\}\\
&=\left\{\{i,i+1\}\ |\ i\in\mathbb{Z}_{n-2}\right\}\ ,
\end{align}
and it is easy to check that the new cost function takes the form
\begin{align}
C'(x)&=1+\sum_{k\in \mathbb{Z}_{n-1}}J_k'x_kx_{k+1}\label{eq:cprimeonedef}
\end{align}
with
\begin{align}
J'_{i}&=\begin{cases}
J_{i}\qquad&\textrm{ when } i\neq n-2\\
J_{n-2}J_{n-1}&\textrm{ when }i=n-2\ 
\end{cases} \label{eq:transformedHx}
\end{align}
We note that the transformation~\eqref{eq:transformedHx} preserves the  parity of the couplings in the sense that
\begin{align}
\prod_{k\in \mathbb{Z}_n} J_k&=\prod_{k\in\mathbb{Z}_{n-1}}J_k'\ .\label{eq:paritypreservation}
\end{align}

Inductively, the $\RQAOA$ thus eliminates variables~$x_{n-1},x_{n-2},\ldots,x_{n_c}$ while imposing the constraints (cf.~\eqref{eq:znminusconstraint})
\begin{align}
x_{n-1}&=x_{n-2}J_{n-2}\\
x_{n-2}&=x_{n-3}J'_{n-3}\\
&\ \ \vdots
\end{align} 
arriving at the cost function~$C_{n_c}(x)$ associated with an Ising chain of length~$n_c$ having couplings belonging to~$\{1,-1\}$.
Because of~\eqref{eq:paritypreservation} and because~\eqref{eq:Hznwithcpl} is frustrated if and only if~$
\prod_{k\in \mathbb{Z}_n} J_k=-1$, we conclude that  
any maximum $x^*\in \{1,-1\}^{n_c}$ of $C_{n_c}(x)$ satisfies
\begin{align}
C_{n_c}(x^*)&=\begin{cases}
n_c+1\qquad& \textrm{ if }\prod_{k\in \mathbb{Z}_n} J_k=1\\
n_c-2\qquad &\textrm{ otherwise}\ .
\end{cases}
\end{align}
Because the cost function acquires a constant energy shift in every variable elimination step (cf.~\eqref{eq:cprimeonedef}), the output $x=\xi(x^*)$ of the $\RQAOA$ algorithm satisfies
\begin{align}
C(x)&=n-n_c+C_{n_c}(x^*)=\begin{cases}
n+1\qquad& \textrm{ if }\prod_{k\in \mathbb{Z}_n} J_k=1\\
n-2\qquad &\textrm{ otherwise}\ .
\end{cases}
\end{align}
This implies the claim.
\end{proof}


\end{document}